\documentclass[preprint]{elsarticle}

\usepackage{array,xspace,multirow,hhline,tikz,colortbl,tabularx,booktabs,fixltx2e,amsmath,amssymb,amsfonts,amsthm}
\usepackage{algorithm}
\usepackage{algorithmic}
\usepackage{verbatim,ifthen}
\usepackage{enumitem}
\usepackage{pifont}
\usepackage{ifthen}
\usepackage{calrsfs,mathrsfs}
\usepackage{bbding,pifont}
\usepackage{pgflibraryshapes}

\usetikzlibrary{trees}
\usetikzlibrary{positioning,chains,fit,shapes,calc}
\usepackage{tkz-graph}

\usepackage{eucal}

\usepackage{tikz}
\usetikzlibrary{arrows}
\usetikzlibrary{decorations.pathreplacing}

	\usepackage{varioref}

\definecolor{light-gray}{gray}{0.9}

	\newtheorem{theorem}{Theorem}%
	\newtheorem{corollary}{Corollary}%
	\newtheorem{example}{Example}
		
        	\newtheorem{fact}{Fact}
		\newtheorem{claim}{Claim}
	        	
	\newcommand\eat[1]{}

	\usepackage{enumitem}
	\setenumerate[1]{label=\rm(\it{\roman{*}}\rm),ref=({\it\roman{*}}),leftmargin=*}
	\newlength{\wordlength}

	\newcommand{\eqclass}[2][]{\ifthenelse{\equal{#1}{}}{[#2]}{[#2]_{\sim_{#1}}}}

	\newcommand{\Pref}[1][]{
		\ifthenelse{\equal{#1}{}}{\mathrel R}{\mathop{R_{#1}}}
	}                                          
	\newcommand{\sPref}[1][]{                  
		\ifthenelse{\equal{#1}{}}{\mathrel P}{\mathop{P_{#1}}}
	}                                          
	\newcommand{\Indiff}[1][]{                 
		\ifthenelse{\equal{#1}{}}{\mathrel I}{\mathop{I_{#1}}}
	}
	\newcommand{\prefset}[1][]{\ifthenelse{\equal{#1}{}}{\mathcal{R}}{\mathcal{R}_{#1}}}

\usepackage{enumitem}
\setenumerate[1]{label=\rm(\it{\roman{*}}\rm),ref=({\it\roman{*}}),leftmargin=*}

\newcommand{\nbh}[1][]{
	\ifthenelse{\equal{#1}{}}{\nu}{\nu(#1)}
}

\newcommand{\cstr}[1][]{
	\ifthenelse{\equal{#1}{}}{\mathscr S}{\cstr(#1)}
}

\newcommand{\choice}[1][]{
	\ifthenelse{\equal{#1}{}}{\mathit{C}}{\choice(#1)}
}

\usepackage{boxedminipage}
\usepackage{xspace}

\newcommand{\pbDef}[3]{%
\noindent
\begin{center}
\begin{boxedminipage}{0.98 \columnwidth}
#1\\[5pt]
\begin{tabular}{l p{0.75 \columnwidth}}
Input: & #2\\
Question: & #3
\end{tabular}
\end{boxedminipage}
\end{center}
}

\tikzset{
  treenode/.style = {align=center, inner sep=0pt, text centered,
    font=\sffamily},
  arn_n/.style = {treenode, circle, white, font=\sffamily\bfseries, draw=black,
    fill=black, text width=1.5em},
  arn_r/.style = {treenode, circle, red, draw=red, 
    text width=1.5em, very thick},
  arn_x/.style = {treenode, rectangle, draw=black,
    minimum width=0.5em, minimum height=0.5em}
}

\begin{document}

\title{Pareto Optimal Allocation under Uncertain Preferences}
	\author{Haris Aziz} \ead{haris.aziz@data61.csiro.au}
\address{Data61, CSIRO and UNSW Australia}
\author{Ronald de Haan} \ead{dehaan@ac.tuwien.ac.at}
\address{Technische Universit\"{a}t Wien, Vienna, Austria.}
\author{Baharak Rastegari} \ead{baharak.rastegari@glasgow.ac.uk}
\address{School of Computing Science, University of Glasgow, Glasgow, UK.}




\begin{abstract}
The assignment problem is one of the most well-studied settings in social choice, matching, and discrete allocation. We consider the problem with the additional feature that agents' preferences involve uncertainty. 
The setting with uncertainty leads to a number of interesting questions including the following ones. How to compute an assignment with the highest probability of being Pareto optimal? What is the complexity of computing the probability that a given assignment is Pareto optimal? 
Does there exist an assignment that is Pareto optimal with probability one?
We consider these problems under two natural  uncertainty models: (1) the lottery model in which each agent has an independent probability distribution over linear orders and (2) the joint probability model that involves a joint probability distribution over preference profiles. For both of the models, we present a number of algorithmic and complexity results highlighting the difference and similarities in the complexity of the two models. 
\end{abstract}

	\begin{keyword}
	 	Assignment Problem, Resource allocation, Pareto optimality, Uncertain Preferences.
		
		\emph{JEL}: C62, C63, and C78
	\end{keyword}

\maketitle

\section{Introduction}

When preferences of agents are aggregated to 
identify a desirable social outcome, Pareto optimality is a minimal requirement. Pareto optimality stipulates that there should not be another outcome that is at least as good for all agents and better for at least one agent. We take Pareto optimality as a central concern and  
consider a richer version of the classic assignment problem where the twist is that agents may express uncertainty in their preferences. The assignment problem is a fundamental setting in which 
$n$ agents express preferences over $n$ items and each agent is to be allocated one item. The setting is a classical one in discrete allocation. Its axiomatic and computational aspects have been well-studied~\citep{AbSo99a,ABL+16a,ACMM05a,AMXY15a,BoMo01a,Gard73b, Sven94a,Sven99a}. Our motivation for studying assignment with uncertain preferences is that agents' preferences may not be completely known because of a lack of information or communication.

Our work is inspired by the recent work of \citet{ABG+16a} who examined the stable marriage problem under uncertain preferences. Uncertainty in preferences  has already been studied in voting~\citep{HAK+12a}. Similarly, in auction theory, it is standard to examine Bayesian settings in which there is probability distribution over the types of the agents. Although computational aspects of Pareto optimal outcomes have been intensely studied in various settings such as assignment, matching, housing markets, and committee voting~\citep{ACMM05a,ALL16a,AMXY15a,AzKe12a, ErEr15a,KMRZ+14a,Manl13a,SeSa13a}, there has not been much work on Pareto optimal under uncertain preferences. When agents have uncertain preferences, one can relax the goal of computing a Pareto optimal outcome and focus on computing outcomes that have the \emph{highest probability} of being Pareto optimal. 
We will abbreviate Pareto optimal as PO. 
If an assignment is Pareto optimal with probability one, we will call it certainly PO. 



    We consider the following uncertainty models:
 
 \begin{itemize}
 \item \textbf{Lottery Model}: For each agent, we are given a probability distribution over linear preferences.
 \item  \textbf{Joint Probability Model}: A probability distribution over linear preference profiles is specified.
 \end{itemize}

 
 Note that both the lottery model and the joint probability model representation can be exponential in the number of agents but if the support of the probability distributions is small, then the representation is compact. 
 Also note that the product of the independent uncertain preferences in the lottery model results in a probability distribution over preference profiles and hence can be represented in the joint probability model. However, the change in representation can result in a blowup. Thus whereas the joint probability model is more general than the lottery model, it is not as compact. In view of this, complexity results for one model do not directly carry over to results for the other model.
 
 The most natural computational problems that we will consider are as follows.
    
\begin{itemize}
    \item {\sc PO-Probability}: what is the probability that a given assignment is PO?
    \item {\sc AssignmentWithHighestPO-Probability}: compute an assignment with the highest probability of being PO.
\end{itemize}
We also consider simpler problems than {\sc PO-Probability}:
\begin{itemize}
    \item {\sc IsPO-ProbabilityNon-Zero}: for a given assignment, is the probability of being PO non-zero?
    \item {\sc IsPO-ProbabilityOne}: for a given assignment, is the probability of being PO one?
\end{itemize}

We also consider a problem connected to {\sc AssignmentWithHighestPO-Probability}:
{\sc ExistsCertainlyPO-Assignment} asks whether there exists an assignment that has probability one for being PO. Note that {\sc ExistsPossiblyPO-Assignment}---the problem of checking whether there exists some PO assignment with non-zero probability---is trivial for all uncertainty models in which the induced `certainly preferred' relation is acyclic. The reason why it is trivial is because the certainly preferred relation can be completed in a way so that it is transitive and then for the completed deterministic preferences, there exists at least one PO assignment.

We say that a given uncertainty model is \emph{independent} if any uncertain preference profile $L$ under the model can be written as a product of uncertain preferences $L_a$ for all agents $a$, where all $L_a$'s are independent~\citep{ABG+16a}. Note that the lottery model is independent but the joint probability model is not.


    \paragraph{Results}
    We show that for both the lottery model and the joint probability model,  {\sc ExistsCertainlyPO-Assignment} is NP-complete. We also prove that {\sc AssignmentWithHighestPO-Probability} is NP-hard for both models. In view of the results, we see that as we move from deterministic preferences to uncertain preferences, the complexity of computing Pareto optimal assignments jumps significantly. 
On the other hand, we show that for a general class of uncertainty models called independent uncertainty models, both problems {\sc IsPO-ProbabilityNon-Zero} and {\sc IsPO-ProbabilityOne} can be solved in linear time. 
Whereas {\sc PO-Probability} is polynomial-time solvable for the joint probability model, we prove that the problem \#P-complete for the lottery model. Even for the lottery model, the problem becomes polynomial-time solvable if there is a constant number of uncertain agents. 

Our results are summarized in Table~\ref{table:summary:uncertainPO}.

     \begin{table*}[h]
         \centering
         \scalebox{0.9}{
         \begin{tabular}{lllll}
             \toprule
          &\textbf{Lottery}&\textbf{Joint Probability}\\
              &\textbf{Model}&\textbf{ Model}\\
	   \textbf{Problems}&&\\
              \midrule
     \multirow{3}{*}{\sc PO-Probability}&\#P-complete  &in P\\
     &but in FPT (parameter &\\
     & \# uncertain agents)&\\
     \midrule
     {\sc IsPO-ProbabilityNon-Zero}&in P&in P\\
     {\sc IsPO-ProbabilityOne}&in P&in P\\
     \midrule
     \multirow{2}{*}{\sc ExistsPossiblyPO-Assignment}&in P&in P\\
     &(trivially exists)&(trivially exists)\\
 \midrule
     {\sc ExistsCertainlyPO-Assignment}&NP-complete&NP-complete\\
     \midrule
      \multirow{1}{*}{\sc AssignmentWithHighestPO-Prob}&NP-hard&NP-hard\\
             \bottomrule
         \end{tabular}
         }

         \caption{Summary of results.} 
         \label{table:summary:uncertainPO}
     \end{table*}

\section{Preliminaries}

The setting we consider is the \emph{assignment problem} which is a triple $(N,O,\succ)$ where $N$ is the set of $n$ agents $\{1,\ldots, n\}$, $O=\{o_1,\ldots, o_n\}$ is the set of items, and $\succ=(\succ_1,\ldots,\succ_n)$ specifies complete, asymmetric, and transitive preferences $\succ_i$ of each agent $i$ over $O$. 
We will denote by $\mathcal{R}(O)$ as the set of all complete and transitive relations over the set of items $O$. We will denote by $\succ_S$ as the preference profile of agents from set $S\subset N$.

An \emph{assignment} is an allocation of items to agents, represented as an
$n\times n$ matrix $[p(i)(o_j)]_{\substack{1\leq i\leq n, 1\leq j\leq n}}$
such that for all $i\in N$, and $o_j\in O$, $p(i)(o_j)\in \{0,1\}$;  and for all $j\in \{1,\ldots, n\}$, $\sum_{i\in N}p(i)(o_j)= 1$. 
An agent $i$ gets item $o_j$ if and only if $p(i)(o_j)= 1$. Each row $p(i)=(p(i)(o_1),\ldots, p(i)(o_m))$ represents the \emph{allocation} of agent $i$.

An assignment $p$ is \emph{Pareto optimal} if there does not exist another assignment $q$ such that $q(i)\succsim_i p(i)$ for all $i\in N$ and $q(i)\succ_i p(i)$ for some $i\in N$.


				We first note a couple of well-known characterisations of Pareto optimal assignments. An assignment $p$ admits a \emph{trading cycle} $o_0,i_0,o_1,i_1,\ldots, o_{k-1},i_{k-1},o_0$ in which $p(i_j)(o_j)>0$ for all $j\in \{0,\ldots, k-1\}$, $o_{j+1 \mod k} \succ_j o_{j\mod k}$ for all $j\in \{0,\ldots, k-1\}$.

 
 
 \begin{fact}[Folklore]\label{fact:sdeff}
     An assignment is Pareto optimal if and only if  it does not admit a trading cycle.
     \end{fact}

We will also use the following characterization of Pareto optimal discrete assignments~\citep{AbSo98a} that is defined with respect to outcomes of serial dictatorship. 
Serial dictatorship is an assignment mechanism that is specified with respect to a permutation $\pi$ over $N$: agents in the permutation are given the most preferred item that is still not allocated. We will denote by $SD(N,O,\succ,\pi)$ the outcome of applying serial dictatorship with respect to permutation $\pi$ over assignment problem $(N,O,\succ)$.    
 
 \begin{fact}[\citet{AbSo98a}]\label{fact:AbSo}
     An assignment is Pareto optimal if and only if it is an outcome of serial dictatorship.
     \end{fact}

Fact~\ref{fact:AbSo} also follows from Proposition 1 by \citet{BrKi05a}.
The facts above show that when preferences are deterministic, a Pareto optimal assignment can be computed or verified easily. We will now focus on similar problems but with the feature that agents have uncertain preferences. 

\begin{example}
	\begin{align*}
		1:&\quad a,b,c ~~(0.6)\\
		&\quad  b,a,c ~~(0.4)\\
		2:&\quad b,a,c\\
		3:&\quad c,b,a
			\end{align*}
			
		Consider the  assignment $abc$ in which $1$ gets $a$, $2$ gets $b$, and $3$ gets $c$. The probability of the assignment being Pareto optimal is 1. On the other hand, the assignment $bac$ has 0.4 probability of being Pareto  optimal.
	\end{example}

    \section{Joint Probability Model}

   
   We first observe that the {\sc PO-Probability} can be solved easily for the joint probability model.
   
              \begin{theorem}\label{thm:joint-POprobability-poly}
                  For the joint probability model, {\sc PO-Probability} can be solved in polynomial time.
                  \end{theorem}
    						              \begin{proof}
                      The probability that a given assignment is PO is equivalent to the probability weight of the preference profiles for which the assignment is PO. This can be checked as follows. We check the preference profiles for which the given assignment is PO (for one profile, this can be checked in linear time). Then we add the probabilities of those profiles for which the assignment is PO. The sum of the probabilities is the probability that the assignment is PO.
                      \end{proof}
                      \begin{corollary}\label{cor:1jointprobab}
                           For the joint probability model, {\sc IsPO-ProbabilityNon-Zero} and {\sc IsPO-ProbabilityOne}  can be solved in polynomial time.
                          \end{corollary}
                      
                          What about {\sc ExistsCertainlyPO-Assignment}?
                          This problem is equivalent to checking whether the sets of PO assignments have a non-empty intersection. We show that this problem is NP-complete even when the probability distribution is over two linear preference profiles.  
		    
		
    		We reduce from the NP-complete problem {\sc Serial\-Dictatorship\-Feasibility}---check whether there exists a permutation of agents for which serial dictatorship gives a particular item $o$ to an agent $i$~\citep{SaSe15a}.

	 \pbDef{{\sc SerialDictatorshipFeasibility}}
	 {$(N,O,\succ, i\in N, o\in O)$}
	 {Does there exist a permutation of agents for which serial dictatorship gives a particular item $o$ to an agent $i$?}

		For linear preference profiles, the set of Pareto optimal allocations are characterized by those that can be achieved via some serial dictatorship. Thus it follows that the following problem is also NP-complete: check whether there exists a Pareto optimal allocation in which a specified agent $i$ gets a specified item $o$.
		 

    		\begin{theorem}\label{joint:npc}
    			For the joint probability model, {\sc ExistsCertainlyPO-Assignment} is NP-complete even when the probability distribution is over two linear preference profiles.  
    			\end{theorem}
    			\begin{proof}
				The problem {\sc ExistsCertainlyPO-Assignment} is in NP because it can be checked in polynomial time whether a given assignment is certainly PO or not (Theorem~\ref{thm:joint-POprobability-poly}).
				
    			To prove NP-hardness, we reduce from the NP-complete problem : {\sc SerialDictatorshipFeasibility} --- given $(N,O,\succ)$, check whether there exists a permutation of agents for which serial dictatorship gives a particular item $o$ to an agent $i$~\citep{SaSe15a}.
				
    				We construct a joint probability over two preference profiles. One of the profiles is the same as $\succ$. In the other preference profile $\succ'$, agent $i$ has $o$ as the most preferred item and has the same order of preference over all other items as in $\succ_i$. Each agent $j\in N\setminus \{i\}$ has $o$ as the least preferred item. As for the other items, each $j\in N\setminus \{i\}$ has the same preferences over the items in $O\setminus \{o\}$ as in $\succ_j$.

  Our first observation is that an assignment is PO under profile $\succ'$ only if  $i$ gets $o$ in it. 
  \begin{claim}
	  An assignment is PO under profile $\succ'$ only if  $i$ gets $o$ in it.
	  \end{claim}
	  \begin{proof}
 The argument is as follows. If $i$ does not get $o$, then an agent $j\neq i$ gets it. However both $i$ and $j$ get a more preferred item under profile $\succ'$ by exchanging their items. 
 \end{proof}
 
   We now prove that we have a yes instance of {\sc SerialDictatorshipFeasibility} if and only if there exists a certainly PO assignment. 
   
	\bigskip	
		
Assume that there exists a certainly PO assignment. Then, it must be PO under $\succ'$ implying that, by our claim above, $i$ gets $o$ in this assignment. The same assignment must also be PO under profile $\succ$ which implies that there exists an assignment that is PO under profile $\succ$ in which $i$ gets $o$. In light of Fact~\ref{fact:AbSo}, this implies that there exists a serial dictatorship the outcome of which under profile $\succ$ is the same assignment. Hence, we have a yes instance of {\sc SerialDictatorshipFeasibility}.
				
%
%

\bigskip
Now consider the case when we have a yes instance of {\sc SerialDictatorshipFeasibility}. This means that there is a permutation $\pi$ under which $i$ gets $o$ when serial dictatorship is run.  Let us call this assignment by $p$. 
Due to Fact~\ref{fact:AbSo}, $p$ is PO under preference profile $\succ$. 
We want to prove that $p$ is PO under each possible preference profile. We already know that it is PO under $\succ$ so it remains to show that it is PO under $\succ'$. 
Due to Fact~\ref{fact:AbSo}, it is sufficient to prove that for profile $\succ'$, there exists a corresponding permutation of agents under which the outcome of serial dictatorship is $p$. 

In fact, we show that for $SD(N,O,\succ',\pi)=p$---i.e., 
the outcome of applying serial dictatorship with permutation $\pi$ is $p$ even if the preference profile is $\succ'$ instead of $\succ$.
In order to prove the statement we prove the following claim.
\begin{claim}
	The following are the same at each round, when applying serial dictatorship to profiles $\succ$ and $\succ'$, in both cases with respect to permutation $\pi$.
	\begin{itemize}
		\item the order in which items are allocated. 
		\item the allocation of each agent.
		\item set of remaining items. 
	\end{itemize}
	\end{claim}
\begin{proof}
The claim can be proved via induction on the number of rounds of serial dictatorship. 
For the base case, let us consider agent $\pi(1)$. 
If $\pi(1)=i$, then $\pi(1)$ picks up $o$ under both preference profiles. This is because, by construction (1) $\pi$ is a permutation under which $i$ gets $o$ when serial dictatorship is applied on $\succ$, and (2) $i$ has $o$ as his most preferred item under $\succ'$.
If $\pi(1)\neq i$, then $\pi(1)$ picks up some item $o'\neq o$ in $p=SD(N,O,\succ,\pi)$. Note that for $\pi(1)$, his most preferred item in both profiles must be $o'$. Hence by the end of the first round, the same item has been given to the same agent in both $\succ$ and $\succ'$.

For the induction, let us assume that $k$ rounds have taken place and the order in which items are allocated, the allocation of each agent in the first $k$ round and the set of unallocated items $T$ is the same under both profiles $\succ$ and $\succ'$. Now consider agent $\pi(k+1)$. 
If $\pi(k+1)=i$, then $i$ picks up item $o$ under $\succ$, implying that $o \in T$, which in turn implies that $i$ must pick up $o$ under $\succ'$ since $o$ is his most preferred item in $O$ under preference $\succ'_i$ and hence his most preferred item in $T$.
It remains to show what happens when $\pi(k+1)\neq i$. In that case $\pi(k+1)$ picks up some item $o'\neq o$ in $SD(N,O,\succ,\pi)$. This means that $o'$ is the most preferred item of agent $\pi(k+1)$  in set $T\subset O$  
under preference profile $\succ$, implying that $o'$ is the most preferred item of agent $\pi(k+1)$ in set $T$ under preference profile $\succ'$ as well. This completes the proof of the claim.
\end{proof}

We have thus proved that the outcome of applying serial dictatorship with respect to permutation $\pi$ is $p$ under both preference profiles $\succ$ and $\succ'$. Thus $p$ is PO under both possibly realizable preference profiles. 
To conclude, we have proved there exists a certainly PO assignment if and only if we have a yes instance of {\sc SerialDictatorshipFeasibility}. Since {\sc SerialDictatorshipFeasibility} is NP-complete, it follows that {\sc ExistsCertainlyPO-Assignment} is NP-complete. 
    \end{proof}

        \begin{corollary}
    	  For the joint probability model,  {\sc AssignmentWithHighestPO-Prob} is NP-hard.
    	    \end{corollary}
        \begin{proof}
    	    Assume to the contrary that there exists a polynomial-time algorithm to solve {\sc AssignmentWithHighestPO-Prob}. In that case, we can compute such an assignment $p$. By Corollary~\ref{cor:1jointprobab}, it can be checked in polynomial time whether $p$ is PO with probability one or not. If $p$ is PO with probability one, then we know that we have a yes instance of {\sc ExistsCertainlyPO-Assignment}. Otherwise, we have a no instance of {\sc ExistsCertainlyPO-Assignment}. Hence {\sc ExistsCertainlyPO-Assignment} is polynomial-time solvable, a contradiction.
\end{proof}
    
    Before dealing with the lottery model, we present some general algorithmic results that apply not just to the lottery model but a class of uncertainty models that includes the lottery model.    
    
	    \section{Independent Uncertainty Models}\label{sec:uncertain}

	    We first present a couple of general results that apply to a large class of uncertainty models that satisfy independence. Recall that a given uncertainty model is \emph{independent} if any uncertain preference profile $L$ under the model can be written as a product of uncertain preferences $L_a$ for all agents $a$, where all $L_a$'s are independent.

	    We first define the \emph{certainly preferred} relation $\succ_i^{certain}$ for agent $i$. We write $b \succ_i^{certain} c$ if and only if agent $i$ prefers $b$ over $c$ with probability 1.

	    \begin{theorem}\label{theorem:POdom}
	        For any independent uncertainty model in which the certainly preferred relation can be computed in polynomial given, given an assignment it can be checked in polynomial-time whether another assignment Pareto dominates it with probability one. 
	        \end{theorem}
	        \begin{proof}
	            Given an assignment $\omega$, we create a trading cycle graph $G$ in which each agent $i$ points to any item $o$ such that $o\succ_i^{certain} \omega(i)$. We now claim that there exists a cycle in $G$ if and only if the assignment $\omega$ is Pareto optimal with probability zero.

	    If there exists a cycle in $G$, then another assignment Pareto dominates $\omega$ with probability one. The reason is that each agent prefers the item he points to with probability one. Hence, if we implement the trade in the cycle, each agent in the cycle gets a certainly more preferred item. Therefore the assignment is Pareto dominated with probability one. 

	    Now suppose that there is an assignment that Pareto dominates $\omega$ with probability one. 
Equivalently, there exists another assignment in which each agent with a different allocation gets a certainly strictly more preferred item. But this means that there exists a cycle in $G$.
	            \end{proof}

%
%
%
%

	        \begin{theorem}\label{theorem:IsPO-ProbabilityOne}
	            For any independent uncertainty model,  {\sc IsPO-ProbabilityOne} can be solved in polynomial time. 
	            \end{theorem}
	            \begin{proof}
	                        Given an assignment $\omega$, we create a trading cycle graph $G$ in which each agent $i$ points to any item $o$ such that $\omega(i)\not\succ_i^{certain} o$. We claim that $\omega$ is Pareto optimal with probability one if and only if $G$ does not contain a cycle.

								
								We first show that if there exists a cycle, then it is not the case that $\omega$ is PO with probability one. Existence of a cycle implies that each agent in the cycle prefers another item to what he has received with non-zero probability, which in turn implies that if we implement the cycle then each of these agents will receive a more preferred item with non-zero probability. Therefore $\omega$ is Pareto dominated with non-zero probability. 

	         
	    	         If it is not the case that $\omega$ is Pareto optimal with probability one, then it must be that that another assignment Pareto dominates it with non-zero probability. Equivalently, there exists another assignment in which each agent with a different allocation gets a different item that is more preferred with non-zero probability. But this means that there exists a cycle in $G$.
	    	         \end{proof}

    \section{Lottery Model}

    We now focus on the lottery model. 
    Since the lottery model is a independent uncertainty model, Theorems~\ref{theorem:POdom} and \ref{theorem:IsPO-ProbabilityOne} apply to it.


    	    \begin{theorem}\label{theorem:IsPO-ProbabilityNonzeroLottery}
    	        For the lottery uncertainty model, {\sc IsPO-ProbabilityNon-Zero} can be solved in polynomial time.
    	        \end{theorem}
    	        \begin{proof}
    	           Consider an assignment $p$ that we want check whether it is PO with non-zero probability. 
We use the following algorithm that can be considered as building a permutation of agents that is consistent with serial dictatorship producing the assignment $p$.

\begin{quote}
Initialize the set of remaining items to $O$, the remaining agents to $N$, and the permutation of the agents $\pi$ to an empty list.
Check if there exists some agent $i$ such that $p(i)$ is an available item that is the most preferred for $i$ in at least one of his preference lists.
If no such agent exists, return no. 
If such an agent exists, give the item to him, append $i$ to the permutation $\pi$, remove $i$ from the set of remaining agents, and remove $p(i)$ from the set of available items. 
Also select the preference of agent $i$ that had $p(i)$ as the most preferred remaining item, denoting it by $\succ_i$. 
Repeat until no more items are left. 
\end{quote}

	If the algorithm builds the whole permutation and does not return no, then we claim $p$ is Pareto optimal with non-zero probability. When an agent $i$ picks the item $p(i)$ in his turn, it means that the agent has at least one possible preference, $\succ_i$, in which $p(i)$ is the most preferred remaining item. Hence when applying serial dictatorship to the selected preference profile $\succ$ with respect to $\pi$, each agent $i$ picks $p(i)$ when his turn comes, resulting in $p$ as the outcome of serial dictatorship, hence implying that $p$ is PO with respect to $\succ$ and therefore PO with non-zero probability.
	
	If the algorithm returns no, we argue that $p$ is PO with zero probability.
Consider the first point in the algorithm where no agent $i$ has $p(i)$ as an available item that is the most preferred for $i$ in at least one of his preference lists. This means that no remaining agent gets his most preferred item (for any preference list) among the available items. Therefore, for each realisation of the preferences profiles, each of the remaining agents is interested in and points to another item held by another agent among the remaining agents. This implies the existence of a trading cycle for each realisation of the preference profiles, where some remaining agents can exchange items among themselves to get a more preferred item than in $p$. Thus $p$ is PO with zero probability.  
 \end{proof}
    
 We now prove that the problem of checking whether there exists an assignment that is PO with probability one is NP-complete.
   Although the proof is similar to the proof of Theorem~\ref{joint:npc}, we give a complete argument since a complexity result for the joint probability model does not directly imply a similar result for the lottery model. 
    
    		\begin{theorem}
    			For the lottery model, {\sc ExistsCertainlyPO-Assignment} is NP-complete.
    			\end{theorem}
    			\begin{proof}

				
							The problem {\sc ExistsCertainlyPO-Assignment} is in NP because it can be checked in polynomial time whether a given assignment is certainly PO or not (Theorem~\ref{theorem:IsPO-ProbabilityOne}). To prove NP-hardness, we use an argument similar to that used in the proof of Theorem~\ref{joint:npc}. 
				
	 We reduce from the NP-complete problem : {\sc SerialDictatorshipFeasibility} --- given an assignment setting $(N, O, \succ)$, check whether there exists a permutation of agents for which serial dictatorship gives a particular item $o$ to an agent $i$~\citep{SaSe15a}.
				
    				We construct preferences in which each agent $j\in N$ has two preference lists where one of them is $\succ_j$. 
				For agent $i$, we add another preference list $\succ_i'$ in which $i$'s most preferred item is $o$ and the rest of the items are in the same order as in $\succ_i$.
For each other agent $j\in N\setminus \{i\}$, we add a preference list $\succ_j'$ which is identical to $\succ_j$ except that $o$ is moved to the end of the list.

Our first observation is that an assignment is PO under profile $\succ'$ only if  $i$ gets $o$ in it. 
If $i$ does not get $o$, and agent $j\neq i$ gets it, then both $i$ and $j$ get a more preferred item under profile $\succ'$ by exchanging their items. 
Hence if there is any assignment that is certainly PO then it must give $o$ to $i$.

We prove that there exists a certainly PO assignment if and only if we have a yes instance of {\sc SerialDictatorshipFeasibility}.

\bigskip
If we have a no instance of {\sc SerialDictatorshipFeasibility}, then in no assignment that is PO under $\succ$ agent $i$ gets $o$. On the other hand, an assignment is PO under $\succ'$ only if $i$ receives $o$. Therefore, there does not exist any certainly PO assignment. 

\bigskip				
Now consider the case when we have a yes instance of {\sc SerialDictatorshipFeasibility}. This means that there is a permutation $\pi$ under which $i$ gets $o$ when serial dictatorship is run.  Let us call this assignment $p$. 
Due to Fact~\ref{fact:AbSo}, $p$ is PO under preference profile $\succ$. 
We want to prove that $p$ is PO under each possible preference profile. 
Due to Fact~\ref{fact:AbSo} it is sufficient to prove that for each possible realizable preference profile, there exists a corresponding permutation of agents under which the outcome of serial dictatorship is $p$. 

In fact, we show that for each possible preference profile $\succ''$, $SD(N,O,\succ'',\pi)=p$ i.e., 
the outcome of applying serial dictatorship with permutation $\pi$ is $p$. 
In order to prove the statement we prove the following claim. (Note that $p$ is PO under $\succ$ with respect to $\pi$.)
\begin{claim}
The following are the same at each round, when applying serial dictatorship to $\succ$ and any of the realizable preference profiles $\succ''$, in both cases with respect to permutation $\pi$
	\begin{itemize}
		\item the order in which items are allocated. 
		\item the allocation of each agent.
		\item set of remaining items. 
	\end{itemize}
	\end{claim}
\begin{proof}

The claim can be proved via induction on the number of rounds of serial dicatorship. 
For the base case, let us consider agent $\pi(1)$. 
If $\pi(1)=i$, then $\pi(1)$ picks up $o$ in all his possible preferences. 
This is because, by construction
(1) $\pi$ is a permutation under which $i$ gets $o$ when serial dictatorship is applied
on $\succ$, so it must be that $i$ ranks $o$ at the top of his list under $\succ_i$ and (2) $i$ has $o$ as his most preferred item under $\succ'_i$ by construction.
If $\pi(1)\neq i$, then $\pi(1)$ picks up some item $o'\neq o$ in $p=SD(N,O,\succ,\pi)$.
Note that for $\pi(1)$, his most preferred item is the same in all possible profiles. Hence by the end of the first round, the same item has been given to the same agent in all the realizable preferences.

For the induction, let us assume that $k$ rounds have taken place and the order in which items are allocated, the allocation of each agent in the first $k$ turns and the set of unallocated items $T$ is the same all the realizable preferences. Now consider the agent $\pi(k+1)$. 
If $\pi(k+1)=i$, then $i$ picks up item $o$ under $\succ$, implying that $o \in T$, which in turn implies that $i$ must pick $o$ under $\succ'_i$ since $o$ is his most preferred item in O under $\succ'_i$ and hence his most preferred item in $T$. 
It remains to show what happens when $\pi(k+1)\neq i$. In that case $\pi(k+1)$ picks some item $o'\neq o$ in $SD(N,O,\succ,\pi)$. This means that $o'$ is the most preferred item of agent $\pi(k+1)$in set $T\subset O$ of agent $\pi(k+1)$
under preference list $\succ_{\pi(k+1)}$, implying that $o'$ is the most preferred item of agent $\pi(k+1)$ in set $T$ under preference $\succ_{\pi(k+1)}'$ as well. 
This completes the proof of the claim.
\end{proof}

We have thus proved that the outcome of applying serial dictatorship with respect to permutation $\pi$ is $p$ under all possible preference profiles. Thus $p$ is PO under each possibly realizable preference profile when we have a yes instance of {\sc SerialDictatorshipFeasibility}. 

To conclude, we have proved there exists a certainly PO assignment if and only if we have a yes instance of {\sc SerialDictatorshipFeasibility}. Since {\sc SerialDictatorshipFeasibility} is NP-hard and {\sc ExistsCertainlyPO-Assignment} is in NP, it follows that {\sc ExistsCertainlyPO-Assignment} is NP-complete. 
    \end{proof}
    
    \begin{corollary}
	    For the lottery model, {\sc AssignmentWithHighestPO-Prob} is NP-hard.
	    \end{corollary}
    \begin{proof}
	    Assume to the contrary that there exists a polynomial-time algorithm to solve {\sc AssignmentWithHighestPO-Prob}. In that case, we can compute such an assignment $p$. By Theorem~\ref{theorem:IsPO-ProbabilityOne}, it can be checked in polynomial time whether $p$ is PO with probability one or not. If $p$ is PO with probability one, then we know that we have a yes instance of {\sc ExistsCertainlyPO-Assignment}. Otherwise, we have a no instance of {\sc ExistsCertainlyPO-Assignment}. 
			Hence {\sc ExistsCertainlyPO-Assignment} is polynomial-time solvable, a contradiction.
	    \end{proof}

In light of Theorem~\ref{theorem:IsPO-ProbabilityNonzeroLottery} and Theorem~\ref{theorem:IsPO-ProbabilityOne}, we know that for the lottery model, it can be checked in polynomial time whether the PO probability of a given assignment is zero or one, respectively. 
  We now turn to the problem of computing the probability that a given assignment is PO. We first present a polynomial-time solution for a restricted setting, and then show that {\sc PO-Probability} is \#P-complete for the lottery model in general.
 
    \begin{theorem}
    \label{thm:lottery-POprobability-poly}
	    For the lottery model, if the number of uncertain agents in constant, then {\sc PO-Probability} is polynomial-time solvable. 
	    \end{theorem}
    \begin{proof}
		Let $\omega$ be a given assignment. Let constant $k$ denote the number of uncertain agents, and let the maximum number of preferences for any uncertain agent be $\ell$. 
			Therefore, the maximum number of preference profiles that are realizable is $\ell^k$ which is still polynomial in the input since $k=O(1)$. For each possible preference profile $\succ$, it is easy to compute the probability of $\omega$ being stable under $\succ$ by simply computing the product of the probabilities of the preferences chosen of the uncertain agents.
			Hence, we have reduced the problem to the problem {\sc PO-Probability} for the joint probability model which can be solved in polynomial time (Theorem~\ref{thm:joint-POprobability-poly}).
	     \end{proof}

    %


\begin{theorem}
	For the lottery model, {\sc PO-Probability} is \#P-complete, even when restricted to the case where each agent has at most two possible preferences.
	\end{theorem}

	\begin{proof}
		We show \#P-hardness by reduction from the \#P-complete problem Monotone-\#2SAT---count the number of satisfying assignments for a 2CNF formula that contains no negation~\citep{Vali79a}. 
		
          	 \pbDef{{\sc Monotone-\#2SAT}}
          	 {A 2CNF formula that contains no negation.}
          	 {Count the number of satisfying assignments.}


		Let $\varphi$ be a monotone 2CNF formula with clauses $c_1,\ldots,c_m$ and variables $x_1,\ldots,x_n$.
		We construct an instance of {\sc PO-Probability} as follows.
		Consider agents $1,\ldots,n$ and items $o_1,\ldots,o_n$, and take the assignment $\sigma$ that gives each agent i item $o_i$.

		We construct the preferences of the agents as follows.
		Take an arbitrary agent $i$.
		Consider the set $\{j_1,\dotsc,j_u\}$ of indices $j$ such that the clause $(x_i \vee x_j)$ occurs in $\varphi$.
		(Without loss of generality, this set $\{j_1,\dotsc,j_u\}$ is non-empty.)
		Suppose that~$j_1 < j_2 < \dotsm < j_u$, in order to fix an
		(arbitrary) order over these indices.
		With probability $\frac{1}{2}$, agent $i$ has $o_i$ at the top of his preference list, followed by the rest of the items in arbitrary order.
		With probability $\frac{1}{2}$, agent $i$ has the following preference: $o_{j_1} \succ_i \dotsm \succ_i o_{j_u} \succ_i o_i \succ_i \dotsm$, where the remaining items appear in arbitrary order after $o_i$.

		This way, the possible preference profiles correspond one-to-one to the possible truth assignments over $x_1,\dotsc,x_n$.
		Namely, taking the preference $o_i \succ_i \dotsm$ for agent $i$ corresponds to setting $x_i$ to 1, and taking the other preference for agent $i$ corresponds to setting $x_i$ to 0.
		Moreover, each possible preference profile occurs with probability $\frac{1}{2^n}$.

	We show that the number of satisfying assignments for $\varphi$ is equal to the number of preference profiles under which $\sigma$ is Pareto optimal. In particular, we show that $\sigma$ is PO under a preference profile if and only if the corresponding truth assignment T satisfies $\varphi$.

		($\implies$) Take a possible preference profile $\succ$ under which $\sigma$ is PO and suppose, for a contradiction, that the corresponding truth assignment T does not satisfy $\varphi$.
		That is, there is some clause $c = (x_i \vee x_j)$ that is not satisfied, implying that in T both $x_i$ and $x_j$ are set to 0.
		Then we know that agent $i$ prefers $o_j$ to $o_i$ and agent $j$ prefers $o_i$ to $o_j$, hence they are willing to swap their assigned items. Therefore $\sigma$ is not Pareto optimal under $\succ$, a contradiction.

		($ \Longleftarrow $)
		Take a possible preference profile $\succ$ and suppose that the corresponding truth assignment T satisfies $\varphi$.
		We show that we cannot find a Pareto improvement of $\sigma$, implying that $\sigma$ is PO.
		Take an arbitrary agent $i$.
		First suppose that T sets $x_i$ to 1.
		This means that agent $i$ prefers $o_i$ to all other items, and so he is not willing to exchange it with another item.
		Now, suppose that T sets $x_i$ to 0. 
		Take the set $\{j_1,\dotsc,j_u\}$ of indices such that the clause $(x_i \vee x_j)$ occurs in $\varphi$. As $x_i$ is set to 0, this means that $i$ prefers $o_{j_1},\dotsc,o_{j_u}$ to $o_i$ and is willing to exchange $o_i$ with either of these items (but no other item). 
		Because T satisfies $\varphi$, we know that T sets $x_{j_1},\dotsc,x_{j_u}$ to 1, and consequently, agents $j_1,\dotsc,j_u$ prefer items $o_{j_1},\dotsc,o_{j_u}$ over all other items (respectively).
		So neither of these agents is willing to exchange their assigned item with $o_i$. 
		Therefore, as no Pareto improvement exists, $\sigma$ is Pareto optimal.

		The number of satisfying truth assignments of $\varphi$ is then exactly equal to $2^n$ times the probability that assignment $\sigma$ is Pareto optimal.
		Thus, {\sc PO-Probability} is \#P-hard, even when restricted to the case where each agent has at most two possible preferences.
		
		Next, we argue that {\sc PO-Probability} is in \#P.
		Technically speaking, the class \#P consists of counting problems,
		which are functions~$f : \Sigma^{*} \rightarrow \mathbb{N}$.
		We can consider {\sc PO-Probability} as such a function producing
		natural numbers in the following way.
		Without loss of generality, suppose that the probabilities in the input
		are all given as rational numbers with the same denominator~$d$.
		(We can transform the input in polynomial time to an equivalent
		input that satisfies this property.)
		Then the probability that the given assignment is Pareto optimal
		is~$\frac{z}{d^n}$ for some positive integer~$z$.
		We then consider the problem {\sc PO-Probability} as the function
		that returns~$z$, rather than the rational~$\frac{z}{d^n}$.
		
		We argue membership in \#P by describing a nondeterministic
		Turing machine~$\mathbb{M}$ that has the property that for each input,
		the number of accepting paths of~$\mathbb{M}$ for this input
		equals the number~$z$ that corresponds to the probability that the
		given matching is Pareto optimal.
		The existence of such a Turing machine implies membership
		in \#P~\cite{Vali79a}.
		The machine~$\mathbb{M}$ operates as follows.
		For each agent~$a_i$, it uses nondeterminism to generate~$d$
		different (partial) computation paths.
		These partial computation paths are concatenated,
		resulting in~$d^n$ total computation paths.
		Suppose that the input specifies~$\ell$ possible preference
		orders for agent~$a_i$, occurring with
		probabilities~$\frac{u_1}{d},\dotsc,\frac{u_{\ell}}{d}$, respectively.
		Then the first~$u_1$ partial computation paths generated for~$a_i$
		correspond to the first preference order, the next~$u_2$
		correspond to the second order, and so on.
		As a result, each total computation path corresponds to some
		preference profile.
		At the end of each computation path, the machine~$\mathbb{M}$
		checks (in deterministic polynomial time) whether the assignment is
		Pareto optimal for the corresponding preference profile, and accepts
		if and only if this is the case.
		It is straightforward to verify that the number of accepting computation
		paths of~$\mathbb{M}$ is exactly the number~$z$ such that the
		probability that the assignment is Pareto optimal is~$\frac{z}{d^n}$.
		Therefore, we know that {\sc PO-Probability} is in \#P.
		\end{proof}

		We showed that
		when there are only a constant number of uncertain agents,
		we can compute the PO probability in polynomial time for the
		lottery model (Theorem~\ref{thm:lottery-POprobability-poly}).
		However, the order of the polynomial that upper bounds the running
		time of our proposed algorithm grows with the number of uncertain
		agents.
		In particular, when~$k$ is the number of uncertain agents,
		and~$\ell$ is the maximum number of possible preference orders
		for these uncertain agents,
		the running time of the algorithm outlined 
		in the proof of Theorem~\ref{thm:lottery-POprobability-poly}
		is~$\Omega(\ell^{k})$.
		We improve on this result by showing that there exists a fixed-parameter
		tractable algorithm that computes the PO probability for the lottery model---%
		that is, an algorithm running in time~$f(k)n^c$
		for some computable function~$f$ and some fixed constant~$c$
		independent of~$k$, where~$n$ denotes the input size.
		In other words, we show that the parameterized problem
		$k$-{\sc PO-Probability}, where the parameter is the number of
		uncertain agents, is fixed-parameter tractable for the lottery model.

		\begin{theorem}
		For the lottery model,
		$k$-{\sc PO-Probability} can be solved in
		fixed-parameter tractable time.
		\end{theorem}
		\begin{proof}
		Take an arbitrary instance of the problem $k$-{\sc PO-Probability},
		consisting of agents~$1,\dotsc,n$, objects~$o_1,\dotsc,o_n$,
		and an assignment $\sigma$.
		Without loss of generality, assume that the
		assignment gives each agent~$i$ the object~$o_i$, and
		that the uncertain agents are agents~$1,\dotsc,k$.
		For each uncertain agent~$i$, let~$\succ_{i,1},\dotsc,\succ_{i,u_i}$
		denote the different possible preferences for agent~$i$.

		Additionally, assume without loss of generality that for
		each of the uncertain agents~$1,\dotsc,k$, each of the possible
		preferences for these agents occurs with probability~$\frac{\ell}{d}$,
		where the numerator~$\ell$ can vary between different agents
		and different possible preferences, but where the
		denominator~$d$ is common among all agents and all possible
		preferences.
		In other words, all probabilities mentioned in the instance are rational
		numbers that share a common denominator~$d$.
		If this were not the case, we could straightforwardly transform the instance
		in polynomial time to an equivalent instance that does satisfy this
		property.

		Also, assume without loss of generality that there exists no
		trading cycle that involves only the agents~$o_{k+1},\dotsc,o_{n}$.
		If this were the case, the assignment is Pareto optimal with probability
		zero, and we can filter out such trivial instances using a polynomial-time preprocessing procedure.

		We now  how to compute the probability that the given assignment
		is Pareto optimal in fixed-parameter tractable time.
		Our computation will proceed in three stages:
		\begin{itemize}
		  \item[(1)] We construct a directed graph~$G$ with~$O(ku2^{k^2})$
		    vertices, where the edges are weighted.
		    Here~$u$ denotes the maximum number of possible preferences
		    for any uncertain agent.
		  \item[(2)] We count the number of homomorphisms~$f$ of a directed
		    path~$P_{2k+2}$ of length~$2k+2$ to this graph~$G$,
		    where each homomorphism is counted multiple times
		    according to (the product of) the weights on the edges in~$f(P_{2k+2})$.
		    This counting can be done in polynomial time
		    using an extension of a known algorithm
		    \cite{FlGr04a,FlGr06a}.
		  \item[(3)] We divide the weighted total number of homomorphisms
		    of~$P_{2k+2}$ to~$G$ by the number~$d^{k}$ to obtain the probability
		    that the given assignment is Pareto optimal.
		\end{itemize}

		We begin with phase~(1), and we construct the weighted,
		directed graph~$G$.
		Let~$\Pi = \{ o_1,\dotsc,o_k \}^2$ be the set of all possible
		pairs~$(o_i,o_j)$ of objects among~$o_1,\dotsc,o_k$.
		We define the set~$V$ of vertices of~$G$ as follows.
		First, we define an auxiliary set~$V'$:
		\[ V' = \{ 1,\dotsc,k+1 \} \cup \{ (i,\succ_{i,j})\ |\ i \in [k+1], j \in [u_i] \}. \]
		Then, we define the set~$V$ of vertices as follows:
		\[ V = \{ s,t \} \cup \{ (v',\Pi')\ |\ v' \in V', \Pi' \subseteq \Pi \}. \]
		That is, the graph~$G$ has vertices~$s$ and~$t$,
		and~$2^{k^2}$ copies of each element in~$V'$
		(one for each~$\Pi' \subseteq \Pi$).
		Intuitively, the vertices~$s$ and~$t$ will act as source and target
		for each homomorphism of~$P_{2k+2}$ to~$G$.

		The sets~$\Pi' \subseteq \Pi$ will intuitively be used to memorize
		the `trading paths' (i.e., paths in the trading cycle graph)
		that result from particular choices of the
		preference orders~$\succ_{i,j}$ chosen for the agents~$1,\dotsc,k$.
		That is, each~$(o_i,o_j) \in \Pi'$ corresponds to a path from~$o_i$
		to~$o_j$ in the directed graph with vertices~$o_1,\dotsc,o_n$
		where there is an edge from~$o_{i'}$ to~$o_{i''}$ if and only if
		agent~$i'$ prefers object~$o_{i''}$ to object~$o_{i'}$.

		We construct the set~$E$ of (weighted and directed) edges as follows.
		\begin{itemize}
		\item
		We add an edge with weight~$1$ from~$s$ to~$(1,\emptyset)$.
		\item
		For each~$i \in [k]$, each~$j \in [u_i]$, and
		each~$\Pi' \subseteq \Pi$, we add an edge from~$(i,\Pi')$
		to~$(i,\succ_{i,j},\Pi')$.
		This edge has weight~$\ell$, where the possible preference
		order~$\succ_{i,j}$ for agent~$i$ occurs with probability~$\frac{\ell}{d}$.
		\item
		For each~$i \in [k]$, each~$j \in [u_i]$, and
		each~$\Pi' \subseteq \Pi$, we add an edge with weight~$1$
		from~$(i,\succ_{i,j},\Pi')$ to the vertex~$(i+1,\Pi'')$,
		for some~$\Pi' \subseteq \Pi'' \subseteq \Pi$.
		The choice of~$\Pi''$ is determined as follows.
		Consider the following graph~$G_{\Pi',\succ_{i,j}}$.
		The vertices of this graph are~$o_1,\dotsc,o_n$.
		For each pair~$(o_{i'},o_{i''})$ of vertices among~$o_{k+1},\dotsc,o_n$,
		there is an edge from~$o_{i'}$ to~$o_{i''}$ if and only
		if agent~$j$ prefers object~$o_{i''}$ to object~$o_{i'}$.
		Moreover, for each~$(o_{i'},o_{i''}) \in \Pi'$, we add an
		edge from~$o_{i'}$ to~$o_{i''}$.
		Finally, for each agent~$o_{i'}$ among~$o_{k+1},\dotsc,o_n$,
		we add an edge from~$o_i$ to~$o_{i'}$ if and only
		if~$o_{i'} \succ_{i,j} o_i$.
		We then let~$\Pi'' \subseteq \Pi$ be the set of all
		pairs~$(o_{i'},o_{i''})$ such that there is a path
		from~$o_{i'}$ to~$o_{i''}$ in~$G_{\Pi',\succ_{i,j}}$.
		Clearly,~$\Pi' \subseteq \Pi''$.
		\item
		For each~$\Pi' \subseteq \Pi$ such that~$(o_i,o_i) \not\in \Pi'$
		for all~$i$ among~$1,\dotsc,k$,
		we add an edge with weight~$1$
		from~$(k+1,\Pi')$ to~$t$.
		\end{itemize}

		Clearly, any homomorphism~$f$ from the directed
		path~$P_{2k+2}$ of length~$2k+2$
		to~$G$ must map the first vertex of the path to~$s$
		and the last vertex of the path to~$t$.
		Each such homomorphism must map
		the $(2i)$-th vertex of the path to some vertex~$(i,\Pi')$
		and the $(2i+1)$-th vertex of the path to some
		vertex~$(i,\succ_{i,j},\Pi')$.
		Also, the $(2k+2)$-th vertex of the path must be
		mapped to some vertex~$(k+1,\Pi')$
		where~$\Pi'$ contains no pair~$(o_i,o_i)$.
		These observations follows directly from the construction of~$G$.

		Moreover, each homomorphism~$f'$ from the directed
		path~$P_{2k+1}$ of length~$2k+1$ to~$G$
		that maps the first vertex of the path to~$s$
		is uniquely determined by some series of
		choices~$\succ_{1,j_1},\dotsc,\succ_{k,j_k}$ for the possible
		preferences of the uncertain agents~$1,\dotsc,k$.
		We argue that such a homomorphism~$f'$ can be extended
		to a homomorphism~$f$ from~$P_{2k+2}$ to~$G$
		if and only if the corresponding
		preferences~$\succ_{1,j_1},\dotsc,\succ_{k,j_k}$
		lead to a trading cycle.
		The homomorphism~$f'$ maps the $(2k+2)$-th vertex of the
		path to some pair~$(k+1,\Pi')$.
		Here~$\Pi'$ is the set of pairs~$(o_i,o_j) \in \{ o_1,\dotsc,o_k \}^2$
		such that the preferences~$\succ_{1,j_1},\dotsc,\succ_{k,j_k}$
		lead to a trading path from~$o_i$ to~$o_j$.
		By our assumption that there exists no trading cycle that involves only the agents~$o_{k+1},\dotsc,o_{n}$,
		we know that the set~$\Pi'$ contains some pair~$(o_i,o_i)$
		if and only if there exists a trading cycle.
		Therefore, by construction of the edges between~$(k+1,\Pi')$
		and~$t$,
		we know that the choices~$\succ_{1,j_1},\dotsc,\succ_{k,j_k}$
		of preferences for the agents~$1,\dotsc,k$ that make the assignment
		Pareto optimal are in one-to-one correspondence with the
		homomorphisms~$f$ from~$P_{2k+2}$ to~$G$.

		We count each such homomorphism~$f$ in a weighted fashion as follows---%
		this is phase~(2).
		Take a homomorphism~$f$ from~$P_{2k+2}$ to~$G$.
		Its weight in the grand total is the product of the weights for each
		edge in~$f(P_{2k+2})$.
		The only edges in~$f(P_{2k+2})$ that have weigth~$>1$
		are edges from~$(i,\Pi')$ to~$(i,\succ_{i,j},\Pi')$.
		Such an edge has weight~$\ell$, where the probability that~$\succ_{i,j}$
		occurs is~$\frac{\ell}{d}$.
		From this, it is straightforward to verify that the total weighted sum of
		all homomorphisms is equal to~$p \cdot d^k$, where~$p$ is the probability
		that the given assignment is Pareto optimal.
		Therefore, in order to compute~$p$, we only need to take the weighted
		sum of all homomorphisms, and divide it by~$d^k$---%
		this is phase~(3) of the algorithm.

		All that remains is to show how we can compute the weighted sum of
		all homomorphisms~$f$ from~$P_{2k+2}$ to~$G$ in polynomial time.
		We can do this by extending a known polynomial-time algorithm to
		count the number of homomorphisms of a graph whose treewidth is bounded
		by a fixed constant into another graph \cite[Theorem~14.7]{FlGr06a}.
		Since paths have treewidth~$1$, counting the number of homomorphisms
		from a path to another graph can be done in polynomial time using this
		algorithm.
		This algorithm uses a dynamic programming approach to count the number
		of homomorphisms.
		This dynamic programming technique can straightforwardly be extended to
		take into account the weights of the homomorphisms.
		(We omit a detailed description of the extended algorithm.)

		This concludes our proof that $k$-{\sc PO-Probability} can be solved in
		fixed-parameter tractable time for the lottery model.
		\end{proof}


    %
    %
    %
    %
    %
    %
    %
    %
    %
    %
    %

		\section{Conclusions}
		
		Computing Pareto optimal outcomes is an active line of research in economics and computer science. In this paper, we examined the problem for an assignment setting where the preferences of the agents are uncertain. Our central technical results are computational hardness results. 
		We see that as we move from deterministic preferences to uncertain preferences, the complexity of computing Pareto optimal outcomes jumps significantly. 
		The computational hardness results carry over to more complex models in which there may be more items than agents, agents may have capacities, and items may have copies. 
		For future work, we are also starting to consider other uncertainty models~\citep{ABG+16a}. 
		If we consider the compact indifference model~\citep{ABG+16a} which is an independent uncertainty model, then the results in Section~\ref{sec:uncertain} apply to it. 
		If we allow for intransitive preferences, even a possibly Pareto optimal assignment may not exist and the problem of checking whether a possible Pareto optimal assignment exists becomes interesting. 
An orthogonal but equally interesting direction will be to consider other fairness, stability, or efficiency desiderata~\citep{Aziz16a}.

 \end{document}